\documentclass{article}

\usepackage{amsmath}
\usepackage{amsthm}
\usepackage{amsfonts}
\usepackage{amssymb,latexsym}
\usepackage{url}
\usepackage{color}

\usepackage{float}
\restylefloat{table}

\theoremstyle{plain}

\newtheorem{thm}{Theorem}
\newtheorem{lem}[thm]{Lemma}

\newtheorem{rem}[thm]{Remark}

\theoremstyle{definition}
\newtheorem{definition}{Definition}

\theoremstyle{remark}

\newcommand{\DU}[2]{_{#1}\Delta_{#2}}
\newcommand{\DF}[2]{_{#1}D_{#2}}

\newcommand{\F}{\mathbb F}
\newcommand{\GF}[1]{\F_{#1}}
\newcommand{\GFb}[1]{\GF{2^{#1}}}
\newcommand{\GFp}[1]{\GF{p^{#1}}}
\def\gf(#1){\GF{#1}}

\newcommand{\Tr}{{\rm Tr}}
\newcommand{\Trn}{{\rm Tr}_n}

\begin{document}

\title{\bf Investigations on $c$-(almost) perfect nonlinear functions}
\author{Sihem Mesnager\thanks{S. Mesnager is with the Department of Mathematics, University of Paris VIII, 93526 Saint-Denis,  with University
Sorbonne Paris Cit\'e, LAGA, UMR 7539, CNRS,
93430 Villetaneuse,  and
also with the T\'el\'ecom Paris, 91120 Palaiseau, France. E-mail: smesnager@univ-paris8.fr}, \protect
\and Constanza Riera\thanks{C. Riera is with Department of Computer Science, Electrical Engineering and Mathematical Sciences,  Western Norway University of Applied Sciences,  5020 Bergen, Norway. %\protect\\
E-mail: csr@hvl.no}, \\
\and Pantelimon~St\u anic\u a\thanks{P. St\u anic\u a is with Applied Mathematics Department, Naval Postgraduate School, Monterey 93943, USA.
E-mail: pstanica@nps.edu} \thanks{Corresponding author},
\and Haode Yan\thanks{H. Yan is with the School of Mathematics, Southwest Jiaotong University, Chengdu, 610031, China. E-mail: hdyan@swjtu.edu.cn},
\and Zhengchun Zhou\thanks{Z. Zhou is with the School of Mathematics, Southwest Jiaotong University, Chengdu, 610031, China.
E-mail: zzc@home.swjtu.edu.cn}}

\maketitle

\begin{abstract}
In a prior paper~\cite{EFRST20}, two of us,
along with P. Ellingsen, P. Felke, and A. Tkachenko,
defined a new (output) multiplicative differential and the corresponding $c$-differential uniformity, which has the potential of extending differential cryptanalysis. Here, we continue the work by looking at some APN functions through the mentioned concept and showing that their $c$-differential uniformity increases significantly in some cases.
\end{abstract}

{\bf Keywords:}
Boolean function,
$p$-ary function,
$c$-differential,
Walsh transform,
differential uniformity,
perfect and almost perfect $c$-nonlinearity.
\newline

{\bf MSC 2000}: 06E30, 11T06, 94A60, 94D10.

\section{Introduction}
Differential cryptanalysis (\cite{BS,BS93}) is one of the most
fundamental cryptanalytic approaches targeting symmetric-key
primitives. This cryptanalysis has attracted a lot of attention
since it was proposed, and it is the first statistical attack for breaking iterated block ciphers~\cite{BS}. The security of cryptographic
functions regarding differential attacks has been widely studied in the
last 30 years. This security is quantified by the so-called
\emph{differential uniformity} of the substitution box (S-box) used in
the cipher~\cite{NK}. A very nice survey on the differential
uniformity of vectorial Boolean functions can be found in the chapter
of Carlet~\cite{CH1}, and the recent book \cite{CarletBook}. Also, an interesting article on this topic is~\cite{Car18}.

Power permutations form a class of suitable candidates since they usually have lower implementation costs in hardware. Their resistance to the standard differential attack attracted attention and was investigated. Besides, several highly nonlinear vectorial functions with low differential uniformity have been derived in the literature. In Table \ref{table-0}, we summarize the known differential spectrum of power functions over $\mathbb{F}_{p^n}$ ($p$ odd). For $p=2$, we refer to \cite[Chapter 11]{CarletBook}.

\begin{table}[H]
\label{table-0}
\footnotesize
\caption{ Known differential spectrum of power functions over $\mathbb{F}_{p^n}$ ($p$ odd).}
\centering
\begin{tabular}{|c||c|c|c|}
%\hline
\hline
%inserts double horizontal lines
Function & conditions & Differential uniform & Refs. \\
[0.5ex]
\hline
$x^{3^n-3}$ & $p=3$, $n>2$ & 5& \cite{XZLH20}\\
\hline
$x^{\frac{p^k+1}{2}}$ & $p$ odd & $\frac{p^{ \mbox{gcd}(n,k)}+1}{2}$ $\mbox { or }$ $p^{ \mbox {gcd}(n,k)+1}$  &\cite{CHNC13}\\
\hline
$x^{\frac{p^n+1}{p^m+1}+\frac{p^n-1}{2}}$ & $ n \mbox { odd}$, $p\equiv 3 \pmod 4$, $m\mid n$ & $\frac{p^m+1}{2} $ &\cite{CHNC13}\\
\hline
$x^{p^{2k}-p^k+1}$ & $p, \frac{n}e \mbox { are odd, gcd} (n,k)=e$ & $p^e+1$ & \cite{Lei2019,Yan_etal19}\\
\hline
\end{tabular}
\end{table}

In \cite{BCJW02}, a new type of differential was introduced that
might be interesting from a practical perspective for ciphers that use
modular multiplication as a primitive operation. This new
differential allowed a modification of the differential cryptanalysis and was used
to cryptanalyze some existing ciphers. The authors argue that one should
look at other types of differential for Boolean vectorial function
$F$, that is, $(F(x+a),cF(x))$ and not only $(F(x+a),F(x))$, as usual. We now describe the differential introduced in~\cite{EFRST20}.

\begin{definition}
  Let  $\GFp n$ denote the field with $p^n$ elements, where $p$ is a prime number and $n$ is a positive  integer. For a function $F:\GFp n\to \GFp n$, and $c\in\GFp n$, we define the (multiplicative) $c$-derivative  of $F$ with respect to $a\in\GFp n$ as
  \begin{equation*}
    \DF ca F(x) = F(x+a) - cF(x),
  \end{equation*}
  for every $x\in\GFp n$. For $b\in\GFp n$, we define
  $\DU cF(a,b)=\#\{ x\in\GFp n,\,\DF ca F(x)=b \}$ and  call
  $\DU cF=\max\{ \DU cF(a,b)\,:\,a,b\in\GFp n,\mbox{ and $a\not=0$ if
    $c=1$}\}$, the $c$-differential uniformity of $F$ (we say that $F$ is $(c,\,\DU cF)$-uniform).
\end{definition}

If the $c$-differential uniformity of $F$ equals $1$, then $F$ is
called a perfect $c$-nonlinear (P$c$N) function. P$c$N functions over
odd characteristic finite fields are also called $c$-planar
functions. If the $c$-differential uniformity of $F$ is $2$, then $F$
is called an almost perfect $c$-nonlinear (AP$c$N) function. It is
easy to see that, for $c=1$, the $c$-differential
uniformity becomes the usual differential uniformity, and P$c$N and
AP$c$N functions become perfect nonlinear (PN) functions, and almost
perfect nonlinear (APN) functions, respectively, play an important
role in both theory and applications. For even characteristics, APN
functions have the lowest differential uniformity. This is, however, not the case for the $c$-differential uniformity, as there exist PcN functions also in even characteristic.

Recently, in an independent work, Bartoli and Timpanella~\cite{BT}   gave a generalization of planar functions, and we recall that below.
\begin{definition}\label{defBT}
Let $\beta \in \mathbb{F}_{p^n} \backslash \{0,1\}$. A function $F : \mathbb{F}_{p^n} \rightarrow \mathbb{F}_{p^n}$ is a $\beta$-planar  function in $\mathbb{F}_{p^n}$ if $\forall~ \gamma \in \mathbb{F}_{p^n},~~ F(x+\gamma) - \beta F(x)$ is a permutation of $\mathbb{F}_{p^n}.$
\end{definition}
In the particular case, when $\beta =-1$, the $\beta$-planar functions are called quasi-planar. In view of the definitions of ~\cite{EFRST20}, the $\beta$-planar functions are simply  PcN functions and quasi-planar functions are PcN functions with $c=-1$.

In \cite{EFRST20}, the authors studied this new notion of differential uniformity for Gold functions and some APN power mappings presented in \cite{HRS99,HS97}. In this paper, we continue the investigations initiated in \cite{EFRST20} and continued in~\cite{BC20,SPRS20,SG20} resolving notably some computational observations presented in the first paper. For the sake of clarity, we display in Table~\ref{table-1} the results presented in this paper, along with some prior ones.

\begin{table}[H]
\label{table-1}
\footnotesize
\caption{Power functions $F(x)=x^d$ over $\gf(p^n)$ with  low $c$-differential uniformity}
\centering
\begin{tabular}{|c||c|c|c|c|}
%\hline
\hline
%inserts double horizontal lines
$p$&$d$ & condition & $_c\Delta_F$ & Refs. \\
[0.5ex]
\hline
% inserts single horizontal line
any& $2$& $c\neq1$ & 2 &\cite{EFRST20}\\
\hline
any &$p^n-2$ &$c=0$ & $1$ &\cite{EFRST20}\\
\hline
2 &$2^n-2$ &$c\neq0$, $\mathrm{Tr_n}(c)=\mathrm{Tr_n}(c^{-1})=1$ & $2$ &\cite{EFRST20}\\
\hline
2 &$2^n-2$ &$c\neq0$, $\mathrm{Tr_n}(c)=0$ or  & $3$ &\cite{EFRST20}\\
& & $\mathrm{Tr_n}(c^{-1})=0$ &&\\
\hline
odd &$p^n-2$ &$c=4,4^{-1}$, or $\eta(c^2-4c)=-1$  & $2$ &\cite{EFRST20}\\
& & and $\eta(1-4c)=-1$ & &\\
\hline
odd &$p^n-2$ &$c\neq0,4,4^{-1}$, $\eta(c^2-4c)$=1   & $3$ &\cite{EFRST20}\\
 &  & or $\eta(1-4c)=1$ &  & \\
\hline
3& $({3^k+1})/{2}$& $c=-1$, $n/\gcd(k,n)=1$ & 1 &\cite{EFRST20}\\
\hline
odd& $({p^2+1})/{2}$& $c=-1$, $n$ odd & 1 &\cite{BT}\\
\hline
odd& $p^2-p+1$& $c=-1$, $n=3$ & 1 &\cite{BT}\\
  \hline
odd & $p^4+(p-2)p^2+$  & $c=-1$, $n=5$ & 1 &\cite{SPRS20}\\
  &   $(p-1)p+1$ &   &  &\\
\hline
odd & $(p^5+1)/(p+1)$ & $c=-1$, $n=5$ &  &\cite{SPRS20}\\
  \hline
  odd & $(p-1)p^6 +p^5+$ & $c=-1$, $n=7$ & 1 &\cite{SPRS20}\\
  &  $(p-2)p^3 +(p-1)p^2 +p$ &  & &\\
    \hline
  odd & $(p-2)p^6 +(p-2)p^5+ $ & $c=-1$, $n=7$ & 1 &\cite{SPRS20}\\
  & $(p-1)p^4 +p^3 +p^2 +p$&  & &\\
  \hline
odd & $(p^7+1)/(p+1)$ & $c=-1$, $n=7$ & 1 &\cite{SPRS20}\\
  \hline
  any & $p^k+1$& $1\neq c\in\gf(p^{\gcd(n,k)})$  & $p^{\gcd(n,k)}+1$ &Thm \ref{pk+1}\\
  \hline
  2& $2^k+1$& $c\neq1$, $\frac{n}{\gcd(n,k)}\geq 3$ ($n$ odd)& $2^{\gcd(n,k)}+1$ &Thm \ref{thm:Gold}\\
  & &  $\frac{n}{\gcd(n,k)}\geq 4$ ($n$ even) & &\\
\hline
odd& $(p^k+1)/2$& $c=1$ & $\leq4$ &  \cite{HS97}\\
\hline
odd& $(p^k+1)/2$& $c=-1$ & $p^{\gcd(n,k)}+1$ &Thm \ref{plusone}\\
\hline
odd& $(p^n+1)/2$& $c\neq\pm1$ & $\leq 4$ &Thm \ref{pn+1over2}\\
  \hline
odd& $(p^n+1)/2$& $c\neq\pm1$, $\eta(\frac{1-c}{1+c})=1$,  & $\leq 2$ &Thm \ref{pn+1over2}\\
& & $p^n\equiv 1 \pmod 4$ & &\\
\hline
%   3  & $\frac{3^n+3}2$ & $c=1$, $n$ odd & 1 & Thm  \ref{3^3+3over2}\\
%  \hline
%  3  & $\frac{3^n+3}2$ & $c=1$, $n$ even & 2 & Thm  \ref{3^3+3over2}\\
%  \hline
%  3  & $\frac{3^n+3}2$ & $c=-1$, $n$ odd & 1 & Thm  \ref{3^3+3over2}\\
%  \hline
  3  & $\frac{3^n+3}2$ & $c=-1$, $n$ even & 2 & Thm  \ref{3^3+3over2}\\
  \hline
 $>3$ & $(p^n+3)/2$&  $p^n\equiv 3 \pmod 4$, $c=-1$ & $\leq 3$ &Thm \ref{pn+3over2}\\
\hline
$>3$& $(p^n+3)/2$&   $p^n\equiv 1 \pmod 4$, $c=-1$ & $\leq 4$ &Thm \ref{pn+3over2}\\
  \hline
   3  & $3^n-3$ &  $c=1, n>1$ odd & $2$ & \cite{HRS99} \\
  \hline
  3  & $3^n-3$ &  $c=1,n>2,n=2\pmod 4$ & $4$ & \cite{XZLH20} \\
  \hline
   3  & $3^n-3$ &  $c=1,n>2,n=0\pmod 4$ & $5$ & \cite{XZLH20} \\
  \hline
   3  & $3^n-3$ &   $c=-1$, $n=0\pmod 4$ & $6$ & Thm \ref{thm:pn-3} \\\hline
  3  & $3^n-3$ &  $c=-1$, $n\not=0\pmod 4$ & $4$ & Thm \ref{thm:pn-3} \\\hline
  3  & $3^n-3$ &   $c=0$ & $2$ & Thm \ref{thm:pn-3} \\
  \hline
  3  & $3^n-3$ &  $c\not=0,-1$ & $\leq 5$ & Thm \ref{thm:pn-3} \\
  \hline
odd& $(p^n-3)/2$& $c=-1$ & $\leq 4$ &Thm \ref{pn-3over2}\\
  \hline
  any & $(2p^n-1)/3$& $p^n\equiv 2 \pmod 3$, $c\neq1$ & $\leq 3$ &Thm \ref{over3}\\
  \hline
\end{tabular}
\end{table}

\section{Notation and Preliminaries}
\label{sec:preliminaries}

Let $n$ be a positive integer and $\F_{p^n}$ denote the  finite field with $p^n$ elements, and $\F_{p^n}^*=\F_{p^n}\setminus\{0\}$ its multiplicative group (for $a\neq 0$, we often write $\frac{1}{a}$ to mean the inverse of $a$ in the multiplicative group). We let $\F_p^n$ be the $n$-dimensional vector space over $\F_p$. We will denote by $\eta(\alpha)$ the quadratic character of $\alpha$ (that is, $\eta(\alpha)=0$ if $\alpha=0$, $\eta(\alpha)=1$ if $0\neq \alpha$ is a square, $\eta(\alpha)=-1$ if $\alpha$ is not a square). $|A|$ will denote the cardinality of a set $A$.
We call a function from $\F_{p^n}$ to $\F_p$  a {\em $p$-ary  function} on $n$ variables.
$\Trn:\F_{p^n}\to \F_p$ is the absolute trace function, given by $\Trn(x)=\sum_{i=0}^{n-1} x^{p^i}$.

Given a $p$-ary  function $f$, the derivative of $f$ with respect to~$a \in \F_{p^n}$ is the $p$-ary  function
$
 D_{a}f(x) =  f(x + a)- f(x), \mbox{ for  all }  x \in \F_{p^n}.
$

For positive integers $n$ and $m$, any map $F:\F_{p^n}\rightarrow\F_{p^m}$ (or, alternatively, $F:\F_p^n\to\F_p^m$, though we will in this paper deal with the former form of the function) is called a {\em vectorial $p$-ary  function}, or {\em $(n,m)$-function}. When $m=n$, $F$ can be uniquely represented as a univariate polynomial over $\F_{p^n}$ (using some identification, via a basis, of the finite field with the vector space) of the form
$
F(x)=\sum_{i=0}^{p^n-1} a_i x^i,\ a_i\in\F_{p^n},
$
whose {\em algebraic degree}   is then the largest Hamming weight of the $p$-ary expansion of the exponents $i$ with $a_i\neq 0$.

\begin{lem}
\label{power}
Let $F(x)=x^d$ be a power function over $\gf(p^n)$. Then
$$
_c\Delta_F=\mathrm{max}\big\{~ \{{_c\Delta_F}(1,b): b\in\gf(p^n) \} \cup \{\gcd(d,p^n-1)\}~\big\}.
$$
\end{lem}
\begin{proof}
For $a=0$ and $c\neq 1$, $_c\Delta_F(0,b)$ is equal to the number of $x\in\gf(p^n)$ such that $x^d=\frac{b}{1-c}$. More precisely,
\begin{eqnarray*}
_c\Delta_F(0,b)=
\begin{cases}
1, &\text{if}~b=0, \\
\gcd(d,p^n-1), &\text{if}~\frac{b}{1-c}\in \gf(p^n)^*~\mathrm{is}~\mathrm{a}~d\mathrm{th~power}, \\
0, &  \text{otherwise}.
\end{cases}
\end{eqnarray*}
Then $\mathrm{max}\{_c\Delta_F(0,b):b\in\gf(p^n)\}=\gcd(d,p^n-1)$.

For $a\neq0$, the two equations $(x+a)^d-cx^d=b$ and $(\frac{x}{a}+1)^d-c(\frac{x}{a})^d=\frac{b}{a^d}$ are equivalent to each other, hence $_c\Delta_F(a,b)={_c}\Delta_F\left(1,\frac{b}{a^d}\right)$. The conclusion then follows from the definition of $_c\Delta_F$.
\end{proof}

We will be using later the next lemma, which is found in~\cite[Lemma 9]{EFRST20}.
\begin{lem}
\label{lem:gcd}
Let $p,k,n$ be integers greater than or equal to $1$ (we take $k\leq n$, though the result can be shown in general). Then
\begin{align*}
&  \gcd(2^{k}+1,2^n-1)=\frac{2^{\gcd(2k,n)}-1}{2^{\gcd(k,n)}-1},  \text{ and if  $p>2$, then}, \\
& \gcd(p^{k}+1,p^n-1)=2,   \text{ if $\frac{n}{\gcd(n,k)}$  is odd},\\
& \gcd(p^{k}+1,p^n-1)=p^{\gcd(k,n)}+1,\text{ if $\frac{n}{\gcd(n,k)}$ is even}.\end{align*}
Consequently, if either $n$ is odd, or $n\equiv 2\pmod 4$ and $k$ is even,   then $\gcd(2^k+1,2^n-1)=1$ and $\gcd(p^k+1,p^n-1)=2$, if $p>2$.
\end{lem}
For more notions and properties of Boolean functions and related notions, the reader can consult~\cite{Bud14,CH1,CarletBook,CS17,MesnagerBook,Tok15}.

\section{Power functions}
\label{sec:power-functions}
In this section, %$\gf(p^n)^\star=\gf(p^n)\setminus\{0\}$,
$\gf(p^n)^\sharp=\gf(p^n)\setminus\{0,-1\}$; recall that $\eta$ denotes the quadratic character on $\gf(p^n)$.
We fix also some notation and list some facts which will be used in this section :
\begin{itemize}
\item $S_{i,j}:=\{x\in \gf(p^n)^\sharp: ~\eta(x+1)=i, \eta(x)=j\}$, where $i,j \in \{\pm 1\}$.
\item $S_{1,1}\sqcup S_{-1,-1}\sqcup S_{1,-1}\sqcup S_{-1,1}=\gf(p^n)^\sharp$ (here `$\sqcup$' denotes the disjoint union).
\item $\Delta_d(x):=(x+1)^d-cx^d$. Trivially, $\Delta_d(0)=1$ and $\Delta_d(-1)=(-1)^{d+1}c$. Where $d$ is clear, we will denote this simply by $\Delta(x)$.
\item $\delta_d(b):=\#\{x\in\gf(p^n):~\Delta_d(x)=b\}$. Where $d$ is clear, we will denote this simply by $\delta(b)$.
%\item $\delta_{i,j}(b)=\#\{x\in S_{i,j}:~\Delta(x)=b\}$, where $i,j \in \{\pm 1\}$.
%\item For $b\neq1,\pm c$, $\delta(b)=\delta_{1,1}(b)+\delta_{-1,-1}(b)+\delta_{1,-1}(b)+\delta_{-1,1}(b)$.
\end{itemize}

\subsection{Gold functions}
\label{sec:gold-functions}
In \cite{EFRST20}, the $c$-differential uniformity of the Gold
function $x\in\GFp n, x\mapsto x^{p^k+1}$ has been studied
in odd characteristic and in even characteristic for $x\in\GFp n, x\mapsto x^3$.
In this section, we push further the study initiated in that paper and
begin with proving that AP$c$N functions can be obtained in any odd characteristic for all $1\neq c\in\gf(p)$.

\begin{thm}
\label{pk+1}
Let $F(x)=x^d$ be a power function over $\gf(p^n)$, where $d=p^k+1$. For $1\neq c\in\gf(p^{\gcd(k,n)})$, the $c$-differential uniformity of $F$ is $_c\Delta_F=\gcd(d,p^n-1)$. Particularly, for $p>2$, if $\frac{n}{\gcd(n,k)}$ is odd, $F(x)$ is AP$c$N, while, if $\frac{n}{\gcd(n,k)}$ is even, $F(x)$ is  $(c,p^{\gcd(k,n)}+1)$-uniform. For $p=2$, $F(x)$ is  $\left(c,\frac{2^{\gcd(2k,n)}-1}{2^{\gcd(k,n)}-1}\right)$-uniform.
\end{thm}
\begin{proof}
Note that
\begin{align*}
\Delta(x)=&(x+1)^d-cx^d\\
=&(1-c)x^{p^k+1}+x^{p^k}+x+1\\
=&(1-c)\left(x^{p^k+1}+\frac{1}{1-c}x^{p^k}+\frac{1}{1-c}x\right)+1\\
=&(1-c)\left(x+\frac{1}{1-c}\right)^{d}+\frac{c}{c-1}.
\end{align*}
The last identity holds since
$\frac{1}{1-c}\in\gf(p^{\gcd(k,n)})$. Then $\Delta(x)$ is a shift of
$x^d$. For $b\in \gf(p^n)$, $\delta(b)\leq\gcd(d,p^n-1)=e$ and there
exists some $b$ such that the equality holds. By Lemma \ref{power},
$F$ is $(c,e)$-uniform. By Lemma \ref{lem:gcd}, $e=\frac{2^{\gcd(2k,n)}-1}{2^{\gcd(k,n)}-1}$ and the theorem is shown.
\end{proof}

Let us specialize our study to the even characteristic. In~\cite{EFRST20}, the odd characteristic case of the Gold function is treated; as to the even characteristic case,
the case $x\mapsto x^3$ is treated (note that there is a typo in~\cite{EFRST20}, the maximal $c$-differential uniformity is 2 for $n\leq 2$ and 3 for $n\geq3$), and there are only some computational results for $3\leq n\leq 8$ and $x\mapsto x^5$.
In this paper, we present a new result that is valid for an infinite family
of values of $n$. This result shows that, surprisingly, the $c$-differential
uniformity of Gold functions may increase significantly (we also recover the claim for $p=2$ of Theorem~\ref{pk+1}).

\begin{thm} 
 \label{thm:Gold}
  Let $2\leq k<n$, $n\geq 3$ and $G(x)=x^{2^k+1}$ be the Gold function
  on $\GFb n$. Assume that $n=md$, where
  $d=\gcd(n,k)$ and $m\geq 3$. If $1\neq c\in  \GFb d$, the $c$-differential uniformity  of $G$ is $\DU cG=\frac{2^{\gcd(2k,n)}-1}{2^{\gcd(k,n)}-1}$. If $c\in\GFb n\setminus \GFb d$, the $c$-differential uniformity  of $G$ is $\DU cG=2^d+1$. 
  \end{thm}
  \begin{proof}
  We consider the differential equation at $a$, say $G(x+a)-c\,G(x)=b$, for some $b\in\F_{2^n}$, which is equivalent to
  \[
  (1-c)\, x^{2^k+1}+ x^{2^k} a+x\, a^{2^k} +a^{2^k+1}-b=0.
  \]
  Dividing by $(1-c)$ and taking $x=y-\frac{a}{1-c}$ this last equation reads
  \begin{equation}
  \label{eq:1}
  y^{2^k+1}+ \frac{a^{2^k}}{1-c} \left(1+\frac{1}{(1-c)^{2^k-1}} \right)y+\frac{c a^{2^k+1}+b (1-c)}{(1-c)^2}=0.
  \end{equation}
  If $1+\frac{1}{(1-c)^{2^k-1}}=0$, which is equivalent to $(1-c)^{2^k-1}=1$, that is, $c\in\F_{p^d}$, Equation~\eqref{eq:1} transforms into 
  \[
   y^{2^k+1}=\frac{c a^{2^k+1}+b (1-c)}{(1-c)^2},
  \]
  which has $ \gcd(2^{k}+1,2^n-1)=\frac{2^{\gcd(2k,n)}-1}{2^{\gcd(k,n)}-1}$ solutions if and only if the right hand side is a $\left(\frac{2^{\gcd(2k,n)}-1}{2^{\gcd(k,n)}-1}\right)$-power (via Lemma~\ref{lem:gcd}). Surely, there are values of $a,b\in\F_{2^n}, b\neq 0$ satisfying that condition; for example, let $a=0$, $b=(1-c)t^{\frac{2^{\gcd(2k,n)}-1}{2^{\gcd(k,n)}-1}}$, for some element $t\in\F_{2^n}^*$. The first claim is shown.
  
 Next, we take $c\in\GFb n\setminus \GFb d$.  Replace $y=\alpha z$ in Equation~\eqref{eq:1}, where $\displaystyle \alpha=\left( \frac{a^{2^k}}{1-c}\left(1+\frac{1}{(1-c)^{2^k-1}} \right)\right)^{2^{-k}}$ (note that the $2^k$-root exists since, for any $p$, $\gcd(p^k,p^n-1)=1$ and $x^{p^k}$ is therefore a permutation on $\F_{p^n}$). The previous equation becomes
  \begin{equation}
  \label{eq:cGold}
  z^{2^k+1}+z+\beta=0,
  \end{equation}
  where $\displaystyle  \beta=\frac{c a^{2^k+1}+b (1-c)}{\alpha^{2^k+1}(1-c)^2}$.

  We will be using some results of~\cite{HK08} (see also~\cite{Bluher04} and \cite{Dob06}).
We first assume that $\gcd(n,k)=1$. By~\cite[Theorem 1]{HK08}, we know that Equation~\eqref{eq:cGold} has either none, one or three solutions in $\F_{2^n}$. In fact, the distribution of these cases for $n$ odd (respectively, $n$ even) is  (denoting by $M_m$ the amount of equations of type ~\eqref{eq:cGold} with $m$ solutions):

  \begin{align*}
  M_0&=\frac{2^n+1}{3}\ \left(\text{respectively}, \frac{2^n-1}{3}\right)\\
  M_1&=2^{n-1}-1\ \left(\text{respectively}, 2^{n-1}\right)\\
  M_3&=\frac{2^{n-1}-1}{3}\ \left(\text{respectively}, \frac{2^{n-1}-2}{3}\right).
  \end{align*}
 Then, for $n\geq3$, $c\neq1$, and $\gcd(n,k)=1$, and since $\beta$ is linear on $b$, this implies that, for any $\beta$ and any $a,c$, we can find $b$ such that $\displaystyle  \beta=\frac{c a^{2^k+1}+b (1-c) }{\alpha^{2^k+1}(1-c)^2}$, so the $c$-differential uniformity of the Gold function is $3$.

  We now assume that $\gcd(n,k)=d>1$. As in~\cite{HK08}, for $v\in\F_{2^n}\setminus \F_{2^d}$, we denote $v_i:=v^{2^{ik}}$, $i\geq 0$, and we define  (for $n=md$) a sequence of polynomials (an approach started by Dobbertin~\cite{Dob04})
  \[
  \begin{array}{l}
  C_1(x)=1;  C_2(x)=1\\
  C_{i+2}(x)=C_{i+1}(x)+x_iC_{i}(x),\ \mbox{ for } 1\leq i\leq n-1.
  \end{array}
  \]
 Let $V=\frac{v_0^{2^{2k}+1}}{(v_0+v_1)^{2^k+1}}$. Then, by ~\cite[Lemma 1]{HK08}
  \[
   C_m(V)=\frac{\Tr^n_d(v_0)}{v_1+v_2}\prod_{j=2}^{m-1} \left( \frac{v_0}{v_0+v_1}\right)^{2^{jk}}.
  \]
   We know by~\cite[Lemma 1]{HK08} that if $n$ is odd (respectively, even) there are $\displaystyle \frac{2^{(m-1)d}-1}{2^{2d}-1}$ (respectively,  $\displaystyle \frac{2^{(m-1)d}-2^d}{2^{2d}-1}$) distinct zeros of $C_m(x)$ in $\F_{2^n}$, that are defined by $V$ above with $\Tr_d^n(v_0)=0$ (also, every zero of $C_m(x)$ occurs with multiplicity $2^k$~\cite{HK08}). Further, by ~\cite[Proposition 4]{HK08}, Equation~\eqref{eq:cGold} has $2^d+1$ zeros in $\F_{2^n}$ for as many as $M_{2^d+1}=\frac{2^{(m-1)d}-1}{2^{2d}-1}$, for $m$ odd, respectively, $M_{2^d+1}=\frac{2^{(m-1)d}-2^d}{2^{2d}-1}$, for $m$ even, values of $\beta$. To be more precise, those $\beta$ achieving this bound must satisfy $C_m(\beta)=0$ (if $d=k$, this is a complete description).  For $m$ odd,  $M_{2^d+1}\geq 1$ is achieved when $m\geq 3$. For $m$ even,  $M_{2^d+1}\geq 1$ is achieved when $m\geq4$. 
   
  The only thing to argue now is whether for a fixed $c\neq 1$, and given $\beta\in\F_{2^n}$, there exist $a,b\in\F_{2^n}$ such that $\displaystyle  \beta=\frac{c a^{2^k+1}+b(1-c)}{\alpha^{2^k+1}(1-c)^2}$, where $\displaystyle \alpha^{2^{k}}=  \frac{a^{2^k}}{1-c}\left(1+\frac{1}{(1-c)^{2^k-1}} \right)$. However, that is easy to see since the obtained equation is linear in $b$.
Therefore, the $c$-differential uniformity of $G$ is $2^d+1$.
\end{proof}

\begin{rem}
If we specialize the above theorem to APN Gold functions, that is when  $\gcd(k,n)=1$, we get
that $F(x)=x^{2^k+1}$ is differentially $(c,3)$-uniform when $1\neq c\in\gf(2^n)$, $n\geq 3$.
\end{rem}

\subsection{Power mappings with low differential $c$-uniformity}
\label{sec:power-mappings-with}
In this subsection, we will use Dickson polynomials of the first kind,
which are defined as
$\displaystyle D_d(x,a)= \sum_{i=0}^{\lfloor \frac{d}{2} \rfloor}
\frac{d}{d-i} \binom{d-i}{ i}(-a)^i x^{d-2i}$, and have the property
$D_m\left(u+\frac{a}{u},a\right)=u^m+\left(\frac{a}{u}\right)^m$ for
$u\in\F_{p^{2n}}$ \cite{LN96}; since, in this paper, the second
variable is always 1,  by abuse, we will
write $D_m(x)$). We will also use Theorem~9 of~\cite{CGM88}, which
states that, for $p$ odd, and supposing $2^r||(p^{2n}-1)$ (where
$2^r||t$ means that $2^r|t$ but $2^{r+1}\nmid t$), then, for
$x_0\in\F_{p^n}$,
$$
|D_d^{-1}(D_d(x_0))|=
\begin{cases}
m, &\mbox{ if }\eta(x_0^2-4)=1,\,D_d(x_0)\neq\pm2\\
\bar{\ell}, &\mbox{ if }\eta(x_0^2-4)=-1,\,D_d(x_0)\neq\pm2\\
\frac{m}{2}, &\mbox{ if }\eta(x_0^2-4)=1,\,2^t||d,\,1\leq t\leq r-2,D_d(x_0)=-2\\
\frac{\bar{\ell}}{2}, &\mbox{ if }\eta(x_0^2-4)=-1,\,2^t||d,\,1\leq t\leq r-2,D_d(x_0)=-2\\
\frac{m+\bar{\ell}}{2}, &\mbox{ otherwise,}
\end{cases}
$$
where $m=\gcd(d,p^n-1),\,\bar{\ell}=\gcd(d,p^n+1)$, and $\eta(\alpha)$ is the quadratic character of $\alpha$.

In \cite{HRS99,HS97},
there are several APN power mappings. Among them, it is shown that
$x^{\frac{5^k+1}2}$ is an APN power function of $\gf(5^n)$ when
$\gcd(2n,k)=1$. In this paper, we study more generally the
$c$-differential uniformity of the power mapping $x^{\frac{p^k+1}{2}}$
in odd characteristic (note that the classical uniformity for this power mapping was studied by Helleseth et al. \cite{HS97} and later by Chou, Hong el al. \cite{CHNC13}) and show:

 \begin{thm}
  \label{plusone}
  Let $p$ be an odd prime, and let $F(x)=x^{\frac{p^k+1}{2}}$ on $\F_{p^n}$, $1\leq k<n$, $n\geq 3$. If $c=-1$, then $F$ is PcN if and only if   $\frac{2n}{\gcd(2n,k)}$ is odd.   Otherwise, $F$ will have the $(-1)$-differential uniformity,  $_{-1}\Delta_F=\frac{p^{\gcd(k,n)}+1}{2}$.
  \end{thm}

  \begin{proof}
  We use the approach of~\cite{CS97,EFRST20}, where it was shown that $x^{\frac{3^k+1}{2}}$ is PcN on $\F_{3^n}$, for $c=\pm 1$. Similarly, our function is $PcN$ if and only if the $c$-derivative $(x+a)^{\frac{p^k+1}{2}}-c x^{\frac{p^k+1}{2}}$ is a permutation polynomial. With a change of variable $z=\frac{-a x}{4}$, we see that this is equivalent to  $\left(z-4\right)^{\frac{p^k+1}{2}}-c z^{\frac{p^k+1}{2}}$ being a permutation polynomial. Shifting by 2, we can see that this happens if and only if $h_c(z)=\left(z-2\right)^{\frac{p^k+1}{2}}-c \left(z+2\right)^{\frac{p^k+1}{2}}$ is a permutation polynomial.
  We can always write $z=y+y^{-1}$, for some $y\in\F_{p^{2n}}$. Our condition (for general $c\neq 1$) becomes
  \allowdisplaybreaks
 \begin{align*}
 h_c(z)&=\left(y+y^{-1}-2\right)^{\frac{p^k+1}{2}}-c\left(y+y^{-1}+2\right)^{\frac{p^k+1}{2}}\\
 &= \frac{\left(y^2-2y+1\right)^{\frac{p^k+1}{2}}-c\left(y^2+2y+1\right)^{\frac{p^k+1}{2}}}{y^{\frac{p^k+1}{2}}}\\
&= \frac{\left(y-1\right)^{p^k+1}-c\left(y+1\right)^{p^k+1}}{y^{\frac{p^k+1}{2}}}\\
 &= \frac{(1-c)y^{p^k+1}-(1+c)y^{p^k}-(1+c) y+(1-c)}{y^{\frac{p^k+1}{2}}}\\
 &=(1-c) y^{\frac{p^k+1}{2}}-(1+c) y^{\frac{p^k-1}{2}}-(1+c) y^{\frac{-p^k+1}{2}}+(1-c) y^{\frac{-p^k-1}{2}}\\
 &= (1-c) D_{\frac{p^k+1}{2}}\left(z\right)-(1+c) D_{\frac{p^k-1}{2}}\left(z\right)
 \end{align*}
 is a permutation polynomial, where $D_m(x)(=D_m(x,1))$ is the Dickson polynomial of the first kind, in our notation.

  If $c=-1$, we obtain that
 $D_{\frac{p^k+1}{2}}(x)$ must be a permutation polynomial, and this is equivalent (by~\cite{No68}, see also \cite{LN96}) to   $\displaystyle \gcd\left( \frac{p^k+1}{2}, p^{2n}-1 \right)=1$. This last identity can be further simplified to $\displaystyle \gcd\left( p^k+1 , p^{2n}-1 \right)=2$. By Lemma~\ref{lem:gcd} a necessary and sufficient condition for that to happen is for   $\displaystyle \frac{2n}{\gcd(2n,k)}$ to be odd (this holds if and only if $k$ is even and, if $n=2^ta,\,2\!\not| a$, with $t\geq0$, and $k=2^\ell b,\,2\!\not| b$, then $\ell\geq t+1$).

In general, if $\ell=\max\left\{\left|D^{-1}_{\frac{p^k+1}{2}}(b)\right|: b\in \F_{p^n}\right\}$, then $_{-1}\Delta_F=\ell$. Here we will use Theorem~9 of~\cite{CGM88} stated above. Here $d=\frac{p^k+1}{2}$, so $m=\gcd\left(\frac{p^k+1}{2},p^n-1\right)=\frac{1}{2}\gcd(p^k+1,p^n-1),\,\bar{\ell}=\gcd\left(\frac{p^k+1}{2},p^n+1\right)=\frac{1}{2}\gcd(p^k+1,p^n+1)$. Then, in general, we can obtain all cases,
and so
\allowdisplaybreaks
\begin{align*}
\ell& =\max\left\{\frac{1}{2}\gcd(p^k+1,p^n-1),\frac{1}{2}\gcd(p^k+1,p^n+1),\right.\\
&\qquad\qquad \left. \frac{1}{4}(\gcd(p^k+1,p^n-1)+\gcd(p^k+1,p^n+1))\right\}\\
&=\max\left\{\frac{1}{2}\gcd(p^k+1,p^n-1),\frac{1}{2}\gcd(p^k+1,p^n+1)\right\}.
\end{align*}
Note that, by Lemma \ref{lem:gcd}, $\gcd(p^k+1,p^n-1)=2$ if $\frac{n}{\gcd(n,k)}$  is odd and $\gcd(p^{k}+1,p^n-1)=p^{\gcd(k,n)}+1$ if $\frac{n}{\gcd(n,k)}$ is even, while
$$\begin{array}{ll}&\gcd(p^k+1,p^n+1)=\gcd(p^k+1,p^n+1-(p^k+1))\\
&=\gcd(p^k+1,p^{k}(p^{n-k}-1))=\gcd(p^k+1,p^{n-k}-1)\\
&=\left\{\begin{array}{l}
2 \mbox{ if }\frac{n-k}{\gcd(n-k,k)}=\frac{n-k}{\gcd(n,k)} \mbox{  is odd}\\
p^{\gcd(k,n)}+1\mbox{ if }\frac{n-k}{\gcd(n-k,k)}=\frac{n-k}{\gcd(n,k)} \mbox{ is even}
\end{array}\right.
 \end{array}$$ so, if $n=2^ta,\,2\!\not| a$, with $t\geq0$, and $k=2^\ell b,\,2\!\not| b$, then, if $\ell\geq t+1$, then $m=1=\bar{\ell}$, and so $\ell=1$ and the function is PcN, while, if $\ell=t$, then $m=2,\,\bar{\ell}=\frac{p^{\gcd(k,n)}+1}{2}=\ell$, while, if $\ell< t$, then $m=\frac{p^{\gcd(k,n)}+1}{2},\,\bar{\ell}=2$, and so $\ell=\frac{p^{\gcd(k,n)}+1}{2}$. Summarizing, if $\ell\leq t$, then  $\ell=\frac{p^{\gcd(k,n)}+1}{2}$.

The proof of the theorem is complete.
\end{proof}

\begin{rem}
  There is a typo in~\textup{\cite{EFRST20}}, where the authors proved that when $p=3$, $d=\frac{3^k+1}{2}$ and $c=-1$, $F(x)=x^d$ is P$c$N over $\gf(p^n)$ if and only if $\frac{2n}{\gcd(2n,k)}$ is odd (not $\frac{n}{\gcd(n,k)}$, as stated there).
  %However, it seems that the condition is not strong enough. If $k$ is odd, $d=\frac{3^k+1}{2}$ is even, then $\gcd(d,3^n-1)\geq2$. By Lemma \ref{power}, the $c$-differential uniformity of $F(x)$ is at least $2$, which is not a P$c$N function.
\end{rem}

Surely, it is natural to wonder  whether the $c$-differential uniformity and
differential uniformity are close or not.  We answer  this question
for the power function $x^{\frac{p^n+1}{2}}$ over $\gf(p^n)$  below
and show that $c$-differential uniformity can be low for $c\neq 1$, although this
power mapping has a high differential uniformity. Note that the proof follows similar lines as those in \cite{HS97}, as do the proofs of Theorems \ref{pn+3over2} and \ref{pn-3over2}.

\begin{rem} 
\label{eta} 
Note that, while $\eta(0)=0$, we have that, for any $y\in\gf(p^n)^*$, $y^\frac{p^n-1}{2}=\eta(y)$, since, given $g$ a generator of $\gf(p^n)$, if $y=g^{2m}$, then $y^\frac{p^n-1}{2}=g^{m(p^n-1)}=1$, and, if $y=g^{2m+1}$, then $y^\frac{p^n-1}{2}=g^{m(p^n-1)}g^\frac{p^n-1}{2}=-1$.  This will be used in the proofs of Theorems~\textup{\ref{pn+1over2}}, \textup{\ref{pn+3over2}} and \textup{\ref{pn-3over2}}.
\end{rem}

\begin{thm}
\label{pn+1over2}
Let $F(x)=x^d$ be a power function over $\gf(p^n)$, where $d=\frac{p^n+1}{2}$ and $p$ is an odd prime. For $\pm1\neq c\in\gf(p^n)$, $_c\Delta_F \leq 4$. Moreover, if $p^n \equiv 1\pmod 4$ and $c$ satisfies $\eta(\frac{1-c}{1+c})=1$, then $_c\Delta_F \leq 2$.
\end{thm}
\begin{proof}

By Remark \ref{eta}, if $y\in\F_{p^n}^*$, then $y^\frac{p^n+1}{2}=y^{\frac{p^n-1}{2}+1}=\eta(y)y$, so, for any $b\in\gf(p^n)$, if $x\in\gf(p^n)^\sharp$ is a solution of $\Delta(x)=(x+1)^d-cx^d=b$, then $x$ satisfies
\begin{equation}\label{eqnpn+1over21}
\eta(x+1)(x+1)-c\eta(x)x=b.
\end{equation}
We distinguish four (disjoint) cases:

\begin{table}[H]
\centering
\begin{tabular}{|l||c|c|c|c|c|}
%\hline
\hline
%inserts double horizontal lines
Set&$\eta(x+1)$&$\eta(x)$&$Equation (\ref{eqnpn+1over21})$ & $x$& $x+1$\\
[0.5ex]
\hline
$S_{1,1}$&1&1&$x+1-cx=b$&$\frac{-b+1}{c-1}$&$\frac{-b+c}{c-1}$\\
\hline
$S_{-1,-1}$&$-1$&$-1$&$-x-1+cx=b$&$\frac{b+1}{c-1}$&$\frac{b+c}{c-1}$\\
\hline
$S_{1,-1}$&1&$-1$&$x+1+cx=b$&$\frac{b-1}{c+1}$&$\frac{b+c}{c+1}$\\
\hline
$S_{-1,1}$&$-1$&1&$-x-1-cx=b$&$\frac{-b-1}{c+1}$&$\frac{-b+c}{c+1}$\\
\hline
\end{tabular}
\end{table}

%Case I. $x\in S_{1,1}$, i.e., $\eta(x+1)=\eta(x)=1$. Then (\ref{eqnpn+1over21}) becomes $x+1-cx=b$, i.e. $x=\frac{-b+1}{c-1}$. That means (\ref{eqnpn+1over21}) has at most one solution in $S_{1,1}$.
%
%Case II. $x\in S_{-1,-1}$, i.e., $\eta(x+1)=\eta(x)=-1$. Then (\ref{eqnpn+1over21}) becomes $-x-1+cx=b$, i.e. $x=\frac{b+1}{c-1}$. That means (\ref{eqnpn+1over21}) has at most one solution in $S_{-1,-1}$.
%
%Case III. $x\in S_{1,-1}$, i.e., $\eta(x+1)=1,\eta(x)=-1$. Then (\ref{eqnpn+1over21}) becomes $x+1+cx=b$, i.e. $x=\frac{b-1}{c+1}$. That means (\ref{eqnpn+1over21}) has at most one solution in $S_{1,-1}$.
%
%Case IV. $x\in S_{-1,1}$, i.e., $\eta(x+1)=-1,\eta(x)=1$. Then (\ref{eqnpn+1over21}) becomes $-x-1-cx=b$, i.e. $x=\frac{-b-1}{c+1}$. That means (\ref{eqnpn+1over21}) has at most one solution in $S_{1,-1}$.
As we see from the table, Equation (\ref{eqnpn+1over21}) has at most one solution in each of the $S_{i,j}$'s.
We have $\delta(b)\leq4$ for $b\neq1,\pm c$ since $\Delta(0)=1$ and $\Delta(-1)=c$ or $-c$. Note that $\Delta(x)=1$ and $\Delta(x)=c$ have no solutions in $S_{1,1}$, and $\Delta(x)=-c$ has no solutions in $S_{1,-1}$, and so, $\delta(0),\delta(\pm c)\leq 4$. Then $_c\Delta_F \leq 4$ follows by Lemma \ref{power} and $\gcd(\frac{p^n+1}{2},p^n-1)\leq2$.

Now we assume that $p^n \equiv 1 \pmod 4$ and $\eta(\frac{1-c}{1+c})=1$, then $\eta(-1)=1$. For fixed $b\neq1,\pm c$, if (\ref{eqnpn+1over21}) has solutions in $S_{1,1}$, then $\eta(\frac{-b+c}{c-1})=1$. If (\ref{eqnpn+1over21}) has solutions in $S_{-1,1}$, then $\eta(\frac{-b+c}{c-1})=-1$. We conclude that (\ref{eqnpn+1over21}) cannot have solutions in $S_{1,1}$ and $S_{-1,1}$ simultaneously. Similarly, we can prove that (\ref{eqnpn+1over21}) cannot have a solution in $S_{-1,-1}$ and $S_{1,-1}$, simultaneously. That means $\delta(b)\leq 2$ for $b\neq1,\pm c$. It can be verified that $\Delta(x)=c$ (respectively, $\Delta(x)=-c$) has no solution in $S_{1,1}$ and $S_{-1,1}$, simultaneously (respectively, in $S_{-1,-1}$ and $S_{1,-1}$), and thus, $\delta(c),\delta(-c)\leq2$. Also, \eqref{eqnpn+1over21} cannot have solutions in both $S_{1,1}$ and $S_{1,-1}$ for $b=1$. Then $\delta(1)\leq2$ since $\eta(-1)=1$, as $\Delta(x)=1$ has no solutions in $S_{-1,-1}$. Hence $_c\Delta_F \leq 2$.
\end{proof}

In \cite{HS97}, the authors studied the differential uniformity of the power functions $x^{\frac{p^n+3}{2}}$ and $x^{\frac{p^n-3}{2}}$ over $\gf(p^n)$. We consider their $c$-differential uniformity when $c=-1$. For $p=3$, the power function $x^\frac{3^n+3}{2}$ is equivalent to $x^\frac{3^{n-1}+1}{2}$, which was studied in Theorem \ref{plusone} when $n$ is odd. We complete those results to the case that $n$ is even in
the theorem below.

\begin{thm}
 \label{3^3+3over2}
 Let $G(x)=x^{\frac{3^n+3}{2}}$ on $\F_{3^n}$, $n\geq 2$. If $c=-1$, then $G$ is APcN if $n$ is even.% If $c=1$, then its differential uniformity is ${_1}\Delta_G=1$, if $n$ is even and ${_1}\Delta_G=4$ if $n$ is odd.
\end{thm}
\begin{proof}
By similar arguments as before, we consider the following (where $z=y+y^{-1}, y\in \F_{3^{2n}}$),
\allowdisplaybreaks
 \begin{align*}
 h_c(z)&=\left(y+y^{-1}-2\right)^{\frac{3^n+3}{2}}-c\left(y+y^{-1}+2\right)^{\frac{3^n+3}{2}}\\
 &= \frac{\left(y-1\right)^{3^n+3}-c\left(y+1\right)^{3^n+3}}{y^{\frac{3^n+3}{2}}}\\
 &= \frac{y^{3^n+3}-y^{3^n}-y^3+1-c\left(y^{3^n+3}+y^{3^n}+y^3+1\right)}{y^{\frac{3^n+3}{2}}}\\
&= \frac{(1-c) \left(y^{3^n+3}+1\right)-(1+c)(y^{3^n}+y^3)}{y^{\frac{3^n+3}{2}}}\\
&=(1-c) D_{\frac{3^n+3}{2}}(z)-(1+c) D_{\frac{3^n-3}{2}}(z).
 \end{align*}

 If $c=-1$, then $h_c(z)= 2 \,D_{\frac{3^n+3}{2}}(z)$. For even $n$,  $\gcd\left(\frac{3^n+3}{2}, 3^{2n}-1 \right)=2$  can be seen from the following argument:
\begin{align*}
\gcd\left(\frac{3^n+3}{2}, 3^{2n}-1 \right)&=\frac{1}{2}\gcd(3^n+3,3^{2n}-1)\\
& =\frac{1}{2}\gcd(3^n+3,3^{2n}-1-(3^n+3)(3^n-1))\\
& =\frac{1}{2}\gcd(3^n+3,2-2\cdot3^n)\\
%&=\frac{1}{2}\gcd(3(3^{n-1}+1),2(1-3^n))\\
& =\frac{1}{2}\gcd(3^{n-1}+1,3^n-1).
\end{align*}
 Now, by Lemma \ref{lem:gcd}, $\gcd(3^{n-1}+1,3^n-1)=3^{\gcd(n-1,n)}+1=4$ since $\frac{n}{\gcd(n-1,n)}=n$ is even, which implies our claim.

Let $n$ be even:
 if $2\leq\ell=\max\left\{\left|D^{-1}_{\frac{{3^n+3}}{2}}(b)\right|: b\in \F_{p^n}\right\}$, then $_{-1}\Delta_F=\ell$. By Theorem~9 of~\cite{CGM88}, supposing $2^r||(3^{2n}-1)$, then, for $x_0\in\F_{3^n}$,
 \begin{equation*}
\label{eq:Dickson}
|D_d^{-1}(D_d(x_0))|=
\begin{cases}
m, &\mbox{ if }\eta(x_0^2-4)=1,\,D_d(x_0)\neq\pm2\\
\bar{\ell}, &\mbox{ if }\eta(x_0^2-4)=-1,\,D_d(x_0)\neq\pm2\\
\frac{m}{2}, &\mbox{ if }\eta(x_0^2-4)=1,\,2^t||d,\,1\leq t\leq r-2,D_d(x_0)=-2\\
\frac{\bar{\ell}}{2}, &\mbox{ if }\eta(x_0^2-4)=-1,\,2^t||d,\,1\leq t\leq r-2,D_d(x_0)=-2\\
\frac{m+\bar{\ell}}{2}, &\mbox{ otherwise,}
\end{cases}
\end{equation*}
where $m=\gcd(\frac{3^n+3}{2},3^n-1)=\frac{1}{2}\gcd(3^n+3,3^n-1)=\frac{1}{2}\gcd(3(3^{n-1}+1),3^n-1)=2,\,\bar{\ell}=\gcd(\frac{3^n+3}{2},3^n+1)=2$, all under $n$ being even. Thus, $\ell=2$.
\end{proof}

In the following, we discuss the $c$-differential uniformity of $x^{\frac{p^n+3}{2}}$ for $c=-1$ and $p>3$.

\begin{thm}
\label{pn+3over2}
Let $F(x)=x^d$ be a power function over $\gf(p^n)$, where $p>3$ is an odd prime and $d=\frac{p^n+3}{2}$. For $c=-1$, $_c\Delta_F \leq 4$ if $p^n\equiv 1 \pmod 4$ and $_c\Delta_F \leq 3$ if $p^n\equiv 3 \pmod 4$.
\end{thm}
\begin{proof}It is easy to see that $\gcd(d,p^n-1)=1$, when $p^n\equiv 3 \pmod 4$ and $\gcd(d,p^n-1)=2$ or $4$ when $p^n\equiv 1 \pmod 4$.
By Remark \ref{eta}, if $y\in\F_{p^n}^*$, then $y^\frac{p^n+3}{2}=y^{\frac{p^n-1}{2}+2}=\eta(y)y^2$, so, for $b\in\gf(p^n)$, if $x\in\gf(p^n)^\sharp$ is a solution of $\Delta(x)=b$, then $x$ satisfies
\begin{equation}\label{eqnpn+3over21}
\eta(x+1)(x+1)^2+\eta(x)x^2=b.
\end{equation}
We distinguish the following (disjoint) four cases.
\begin{table}[H]
\centering
\begin{tabular}{|l||c|c|c|c|}
%\hline
\hline
%inserts double horizontal lines
&Set&$\eta(x+1)$&$\eta(x)$&$Equation (\ref{eqnpn+3over21})$ \\
[0.5ex]
\hline
Case I&$S_{1,1}$&1&1&$x(x+1)=\frac{b-1}{2}$\\
\hline
Case II&$S_{-1,-1}$&$-1$&$-1$&$x(x+1)=\frac{-b-1}{2}$\\
\hline
Case III&$S_{1,-1}$&1&$-1$&$x=\frac{b-1}{2}$, $x+1=\frac{b+1}{2}$\\
\hline
Case IV&$S_{-1,1}$&$-1$&1&$x=\frac{-b-1}{2}$, $x+1=\frac{-b+1}{2}$\\
\hline
\end{tabular}
\end{table}

%Case I. $x\in S_{1,1}$, i.e., $\eta(x+1)=\eta(x)=1$. Then (\ref{eqnpn+3over21}) becomes $x(x+1)=\frac{b-1}{2}$, which has at most two solutions.
%
%Case II. $x\in S_{-1,-1}$, i.e., $\eta(x+1)=\eta(x)=-1$. Then (\ref{eqnpn+3over21}) becomes $x(x+1)=\frac{-b-1}{2}$, which has at most two solutions.
%
%Case III. $x\in S_{1,-1}$, i.e., $\eta(x+1)=1, \eta(x)=-1$. Then we obtain $x=\frac{b-1}{2}$ from (\ref{eqnpn+3over21}).
%
%Case IV. $x\in S_{-1,1}$, i.e., $\eta(x+1)=-1, \eta(x)=1$. Then we obtain $x=\frac{-b-1}{2}$ from (\ref{eqnpn+3over21}).

First, we assert that $\Delta(x)=b$ cannot have solutions in $S_{1,1}$ and $S_{1,-1}$, simultaneously, for fixed $b\in\gf(p^n)$. Suppose, on the contrary, then $\eta(\frac{b-1}{2})=-1$ since it is a solution in $S_{1,-1}$, and $\eta(\frac{b-1}{2})=1$ since \eqref{eqnpn+3over21} has solutions in $S_{1,1}$, which is a contradiction.

If $p^n\equiv 1 \pmod 4$, then $\eta(-1)=1$. Then $\Delta(x)=b$ cannot have solutions in $S_{1,1}$ and $S_{-1,1}$ simultaneously. Otherwise, we obtain $\eta(\frac{-b+1}{2})=-1$ from Case~IV and $\eta(\frac{b-1}{2})=1$ from Case~I, which is a contradiction.  This means that for any $b$, the solutions of (\ref{eqnpn+3over21}) in $\gf(p^n)^\sharp$ is at most $4$. It is easy to see that $\Delta(0)=\Delta(-1)=1$, since $d$ is even. For $b=1$, it can be verified that $\Delta(x)=1$ has no solution in $S_{1,1},S_{1,-1}$ and $S_{-1,1}$, and so, $\delta(1)\leq4$. This, along with $\gcd(d,p^n-1)=2$ or $4$ implies that $_c\Delta_F\leq4$.

If $p^n\equiv 3 \pmod 4$, then $\eta(-1)=-1$. If $x_1\in S_{1,1}$ is a solution of $x(x+1)=\frac{b-1}{2}$, then the other solution is $-x_1-1$. Note that $x_1$ and $-x_1-1$ cannot be in $S_{1,1}$ simultaneously, so (\ref{eqnpn+3over21}) has at most $1$ solution in $S_{1,1}$. Similarly, (\ref{eqnpn+3over21}) has at most $1$ solution in $S_{-1,-1}$. Then $\Delta(x)=b$ has at most $3$ solutions in $\gf(p^n)^\sharp$ since $\Delta(x)=b$ cannot have solutions in both $S_{1,1}$ and $S_{1,-1}$. It is clear that $\Delta(0)=1$ and $\Delta(-1)=-1$, since $d$ is odd. For $b=1$ and $b=-1$, it can be verified that $\Delta(x)=1$ and $\Delta(x)=-1$ have no solution in $\gf(p^n)^\sharp$. We conclude that $_c\Delta_F\leq3$.
\end{proof}

It has been shown that $x^{3^n-3}$ is an APN power mapping of
$\gf(3^n)$ when $n>1$ is odd (see \cite{HS97,HRS99}). In \cite{XZLH20}, it was shown that the differential uniformity of $x^{3^n-3}$ over
$\gf(3^n)$ is respectively 4 or 5 when $n=2\pmod 4$ or $n=0\pmod 4$. We show in the next
theorem that this power mapping (which is equivalent to $x^{3^{n-1}-1}$) has also low $c$-differential
uniformity for $c\in\{0,-1\}$.

  \begin{thm}
  \label{thm:pn-3}
  Let $p=3$, $n\geq 2$ and $F(x)=x^{3^n-3}$ on $\F_{3^n}$.  If $c=-1$, the $c$-differential uniformity of  $F$ is $6$ for  $n\equiv 0\pmod 4$ and $4$, otherwise.
  If $c=0$,  the $c$-differential uniformity of $F$ over $\F_{3^n}$ is $2$.
   If $c\neq 0,\pm1$, the $c$-differential uniformity of $F$  is
  $\leq 5$. Moreover, the   $c$-differential uniformity of $4$ is attained for some $c$, for all $n\geq 3$, and the $c$-differential uniformity of $5$ is attained (at least) for all positive $n\equiv 0\pmod 4$.
  \end{thm}
  \begin{proof}
  For $a,b\in\F_{3^n}$, we look at the equation $F(x+a)-cF(x)=b$, that is,
  \begin{equation}
  \label{eq:thm6}
  (x+a)^{3^n-3}- c\, x^{3^n-3}=b.
  \end{equation}
  Note that, for $y=0$, $y^{3^n-3}=0$, while, if $y\in\F_{3^n}^*$, then $y^{3^n-3}=y^{-2}$. This fact will be used in several parts of the proof.

  If $c=0$, the equation is then $(x+a)^{3^n-3}=b$. If $a=b=0$, then we get the unique solution $x=0$. If $a=0,b\neq 0$, the equation is then $bx^2=1$, which has two solutions if $b$ is a square and none, otherwise. If $a\neq 0,b=0$, the solution is $x=-a$. If $ab\neq 0$, the equation is then $b(x+a)^2=a$, which has two solutions if $a/b$ is a nonzero perfect square (always realizable). Summarizing, the function is APcN with respect to $c=0$. For the remainder of the proof, we assume that $c\neq 0,1$.

  If $a=b=0$, then $x=0$ is the only solution for Equation~\eqref{eq:thm6}.
  If $a=0$ and $b\neq 0$, then Equation~\eqref{eq:thm6} becomes $(1-c) x^{3^n-3}=b$. Surely, $x\neq 0$ and so, the equation becomes $bx^2=(1-c)$. Taking $\alpha\neq 0$ a fixed element of $\F_{3^n}$ and $b=\frac{1-c}{\alpha^2}$, then $x=\pm \alpha$ are solutions for this equation (observe that if $b\neq 0$ and $\frac{1-c}{b}$ is not a square in $\F_{3^n}$ there are no solutions for $(1-c) x^{3^n-3}=b$; for $c$ fixed, there are  $\frac{3^n-1}{2}$ nonzero values of $b$ such that $\frac{1-c}{b}$ is a square in $\F_{3^n}$).

  If $a\neq 0$ and $b=0$, then $0\neq x\neq -a$, and Equation~\eqref{eq:thm6} becomes $(x+a)^{-2}= cx^{-2}$, that is, $\left(\frac{x}{x+a}\right)^2=c$, which has two solutions depending on whether $c$ is a nonzero perfect square or not (there are $\frac{3^n+1}{2}$ such $c$'s).

  When $ab\neq 0$, we observe that $x=0$, $x=-a$ are solutions of Equation~\eqref{eq:thm6} if and only if $b=a^{-2}$, respectively, $b=-c a^{-2}$.
  Note that these can happen simultaneously if and only if $c=-1$.

  We now assume $ab\neq 0$, $x\neq 0,-a$.
 Equation~\eqref{eq:thm6} is therefore
 \allowdisplaybreaks
 \begin{align}
 & (x+a)^{-2}-c x^{-2}=b,\text{ that is,}\nonumber\\
 & b (x+a)^2 x^2 -x^2+c(x+a)^2=0,\text{ that is,}\nonumber\\
 & b x^4-b a x^3+b a^2 x^2-x^2+c x^2-ac x+c a^2=0,\nonumber\\
 &x^4-a  x^3+\frac{ba^2+c-1}{b}x^2-\frac{ac}{b}x+\frac{ca^2}{b}=0.\label{eq:thm6-quartic}
  \end{align}

   Using $x=y+a$ in~\eqref{eq:thm6-quartic}, we obtain
  \begin{equation}
  \label{eq:main7}
  y^4+\frac{b a^2+c-1}{b}\, y^2+\frac{a(c+1)}{b}\, y
 +\frac{b a^4+(c-1) a^2}{b}=0.
  \end{equation}
    Now, if $c=-1$ and $b=a^{-2}$, Equation~\eqref{eq:main7} is then
  \begin{equation}
  \label{eq:main7_1}
 y^4-a^2 y^2-a^4=0,
  \end{equation}
  whose discriminant (of the underlying quadratic) is $2a^4$. Now, we know that the underlying quadratic of Equation~\eqref{eq:main7_1} has two distinct roots if and only if the above discriminant is a nonzero square. It is undoubtedly nonzero for $a\neq 0$, and we know that $2=-1$ is a perfect square, say $-1=\imath^2$ in $\F_{3^n}$ if and only if $n$ is even. For the quartic to have four distinct roots in $\F_{3^n}$, obtained from $y^2= a^2(1\pm \imath)$ (these are the roots of the underlying quadratic), one needs $1\pm \imath$ to be a perfect square, which happens if $n\equiv 0\pmod 4$, and we provide the reason next.
The minimal polynomial of one of the roots, say $\sqrt{1+\imath}$, is $x^4+x^2-1$ which has the roots $\{\alpha^2 + 2, 2\alpha^3 + 2\alpha^2 + 2\alpha,2\alpha^2 + 1,\alpha^3 + \alpha^2 + \alpha \}$ in $\F_{3^4}$ (we took the primitive polynomial $x^4-x^3-1$ with $\alpha$ as one of the roots), and consequently in all extensions $\F_{3^n}$ of $\F_{3^4}$ with $n\equiv 0\pmod 4$, and no others. Further, if $n\equiv 2\pmod 4$ we do not get more roots from Equation~\eqref{eq:main7_1} under $c=-1,b=a^{-2}$, therefore, in this case we only get the solutions $x=0,-a$.

    If $c=-1$ and $b\neq a^{-2}$, we let $b=a$, obtaining the equation  $a^2 + a^5 + (1 + a^3) y^2 + a y^4=0$. The  discriminant of the underlying quadratic is $1+a^3=(1+a)^3$. Thus, if   $a=d^4-1$ for some random fixed $d\neq0$ such that $a\neq 0,1$, the previous equation has four solutions and so, the first claim is shown.

  If $c\neq \pm 1$, putting together the  potential solutions for Equation~\eqref{eq:main7} and $y=0$ or $y=-a$, we see that we cannot have more than $5$ solutions.

  We now show that the differential uniformity of $4$ is attained.
 For $a=1,b=-1$,
 the quartic~\eqref{eq:main7} becomes
 $
y^4-(1+c) y^2-(1+c)y-(1+c)=0,
 $
 which is equivalent to $y^4-(c+1) (y^2+y+1)=0$, and further, $y^4-(c+1)(y-1)^2=0$ (observe that $y\neq 0,1$).
 Thus, taking $c=\alpha^2-1$, for some $\alpha\neq 0,1$, then $\frac{y^2}{y-1}=\pm \alpha$, which is equivalent to
 \[
 y^2\mp\alpha y\pm \alpha=0.
 \]
 These pairs of equations will each have two roots  if and only if
 $\eta(\alpha^2\pm 4\alpha)=\eta(\alpha^2\pm \alpha)=1$ ($\eta$ is the quadratic character), that is, if $\alpha^2\pm \alpha$ are nonzero squares in $\F_{3^n}$.
 However, we do have an argument for the existence of such $\alpha$, in general. We use~\cite[Theorem 5.48]{LN96}, which states that if $f(x)=a_2 x^2+a_1 x+a_0$ is a polynomial in a finite field $\F_q$ of odd characteristic, $a_2\neq 0, d=a_1^2-4a_0a_2$, and $\eta$ is the quadratic character on $\F_q$, then the Jacobsthal sum
 \[
 \sum_{x\in\F_q} \eta(f(x))=\begin{cases}
 -\eta(a_2)&\text{ if }  d\neq 0\\
 (q-1)\eta(a_2)&\text{ if }  d= 0.
 \end{cases}
 \]
 First, we take $f(x)=x^2-x$, $d=1$, and so, the Jacobsthal sum becomes
 \begin{align*}
&   \sum_{x\in\F_{3^n}} \eta(x^2-x)=-\eta(1)=-1,\text {so,}\\
&-1=\eta(0)+\eta(1^2-1)+\sum_{x\in\F_{3^n},x\neq 0,1} \eta(x^2-x)=\sum_{x\in\F_{3^n},x\neq 0,1} \eta(x^2-x).
 \end{align*}
 Thus,
 $\displaystyle
 \sum_{x\in\F_{3^n},x\neq 0,1} \eta(x^2-x)=-1,
 $ and similarly,
  $\displaystyle
 \sum_{x\in\F_{3^n},x\neq 0,-1} \eta(x^2+x)=-1.
 $
   Let $N_1=|\{x|x\neq 0,\pm 1, \eta(x^2-x)=1\}$  and $N_2=|\{x|x\neq 0,\pm 1, \eta(x^2+x)=1\}$.
 We compute
 \begin{align*}
  -1&=\sum_{x\in\F_{3^n},x\neq 0,1} \eta(x^2-x)\\
  &=\eta((-1)^2-(-1))+\sum_{x\in\F_{3^n},x\neq 0,\pm 1} \eta(x^2-x)\\
  &=\eta(2)+\sum_{x\in\F_{3^n},x\neq 0,\pm 1} \eta(x^2-x)\\
  &=\eta(-1)+\sum_{x\in\F_{3^n},x\neq 0,\pm 1} \eta(x^2-x)=\eta(-1)+N_1-(3^n-3-N_1)\\
  -1&= \sum_{x\in\F_{3^n},x\neq 0,-1} \eta(x^2+x)\\
  &=\eta(1^2+1)+\sum_{x\in\F_{3^n},x\neq 0,\pm 1} \eta(x^2+x)\\
  &= \eta(-1)+\sum_{x\in\F_{3^n},x\neq 0,\pm 1} \eta(x^2+x)=\eta(-1)+N_2-(3^n-3-N_2).
 \end{align*}
 We therefore get
 $
 N_1=N_2=\frac{3^n-4-\eta(-1)}{2}.
 $
 The sets of $x\neq 0,\pm 1$ of cardinality $ \frac{3^n-4-\eta(-1)}{2}$ (equal to $\frac{3^n-3}{2}$ for $n$ odd and  $\frac{3^n-5}{2}$ for $n$ even) such that $x^2-x$, respectively, $x^2+x$ are squares, may still be disjoint. If they are not disjoint, we are done. Suppose now that the sets are disjoint; then,  we need to consider the Jacobsthal sum (again, using~\cite[Theorem 5.48]{LN96})
 \allowdisplaybreaks
 \begin{align*}
 & \sum_{x\in\F_{3^n},x\neq 0,\pm 1} \eta(x^2-x)\eta(x^2+x)
= \sum_{x\in\F_{3^n},x\neq 0,\pm 1} \eta((x^2-x)(x^2+x))\\
  &= \sum_{x\in\F_{3^n},x\neq 0,\pm 1}\eta(x^2(x^2-1))
   = \sum_{x\in\F_{3^n},x\neq 0,\pm 1} \eta(x^2)\eta(x^2-1)\\
  &= \sum_{x\in\F_{3^n},x\neq 0,\pm 1} \eta(x^2-1)
  =\sum_{x\in\F_{3^n}} \eta(x^2-1)-\eta(-1)=-1+\eta(-1),
 \end{align*}
 but that is impossible for $n\geq 3$, if $x^2-x$ and $x^2+x$ are never squares at the same time (given our prior counts, $N_1,N_2$).

 We now take $n$ to be an even integer.  We shall show that there are values of $c$ such that the $c$-differential uniformity is $5$. We take $b=a^{-2}$, so that $x=0$ is a solution of Equation (\ref{eq:thm6}).
 Under this condition,  the idea is to find, $A,B$ such that the quartic in $y$ can be written as $(y^2+1)^2+A (y+B)^2=0$, where $B$ is a perfect square.
 We let $a$ such that $a^6+a^2+2=0$. This equation has the solutions $g+1,2g+2$  in $\F_{3^2}$, where $g$ is the primitive root vanishing $X^2-X-1=0$, and consequently, by field embedding, it is solvable in $\F_{3^n}$ for all even $n$.  Further, we take  $c=da^{-2}$, where $d=\frac{2 +2 a^6}{1 +2 a^2 + 2 a^4}$. Our quartic in $y$ becomes
 $
 y^4+ d y^2 + (a^3 + a d) y  +a^2 d=0,
 $
 which we will write in the form
 \[
 (y^2 + 1)^2 + (1 + d) \left(y -\frac{ a^3+ad}{ 1+d}\right)^2=0.
 \]
 Observe that $\frac{ a^3+ad}{ 1+d}=-\frac{a^4+1}{a^3-a}$, when $d$ has the value we previously chose.
 Thus, assuming that $-(1+d)=-\frac{1 + a^4}{2 + a^2 + a^4}$ is a perfect square, say $\beta^2$ (we will check if that happens later),
 the equations will give the four solutions (we let $\imath$ be the solution to $X^2=-1$ in $\F_{3^2}$, and any other even extension of $\F_{3^2}$)
 \[\
 y^2+1= \pm\beta  \left(y +\frac{ a^4+1}{ a^3-a}\right),\text{ that is, }
 y^2\mp \beta y\mp \beta \frac{ a^4+1}{ a^3-a}+1=0.
\]
 These last equations will have two solutions each if and only if
 $\displaystyle
 \displaystyle \beta^2\pm   \beta \frac{ a^4+1}{ a^3-a}-1
 $
  is a perfect nonzero square. Using SageMath, we  quickly found values of $a$ satisfying the equation $a^6+a+2=0$ in $\F_{3^4}$  such that $-(1+d)=-\frac{1 + a^4}{2 + a^2 + a^4}$ and the expressions above are perfect squares, as well and, by field embedding, there are values of $a$ for every dimension divisible by $4$ where the last displayed expression is a square, as well.
  The theorem is shown.
   \end{proof}
   
   \begin{rem}
One referee suggested some alternative arguments for the fact  that there exists $\alpha\in\F_{p^n}^*$
for which both $\alpha^2-\alpha$ and $\alpha^2+\alpha$ are squares, in the latter part of our proof. Surely, $x^2+x=y^2$ and $x^2-x=z^2$
is equivalent to $x=(z^2-y^2)/2$ and $(z^2-y^2)^2/4+(z^2-y^2)/2=y^2$.  This last equation defines an absolutely irreducible quartic curve and therefore
there exist triples $(x,y,z)\in\F_{p^n}^3$ such that $x^2+x=y^2$ and $x^2-x=z^2$ and $xyz\neq0$. Our argument provides some idea on counts, as well.
\end{rem}

   \begin{rem}
  Our computations in SageMath revealed that there are other  values of the $c$-differential uniformity for the function in Theorem~\textup{\ref{thm:pn-3}}. In fact, if $c\neq \pm 1$ and $n=2$, then ${_c}\Delta_F=2;$  when $n=3$, we have ${_c}\Delta_F\in\{3,4\};$ for $n=4$, then ${_c}\Delta_F\in\{2,4,5\};$  for $n=5$, we  have ${_c}\Delta_F\in\{4\};$ if $n=6$, then ${_c}\Delta_F\in\{4,5\}$.
  \end{rem}

We finally present in the following two theorems two power mappings having low $c$-differential uniformity.

\begin{thm}\label{pn-3over2}Let $F(x)=x^d$ be a power function over $\gf(p^n)$, where $p$ is an odd prime and $d=\frac{p^n-3}{2}$. For $c=-1$, $_c\Delta_F \leq 4$.
\end{thm}
\begin{proof}It is easy to see that $\gcd(d,p^n-1)=1$, when $p^n\equiv 1 \pmod 4$ and $\gcd(d,p^n-1)=2$ when $p^n\equiv 3 \pmod 4$. For $b\in\gf(p^n)$, we consider the equation
\begin{equation}\label{eqnpn-3over21}
\Delta(x)=(x+1)^d+x^d=b.
\end{equation}
If $b=0$, (\ref{eqnpn-3over21}) has the unique solution $x=-\frac{1}{2}$, when $p^n\equiv 1 \pmod 4$, and has no solution, when $p^n\equiv 3 \pmod 4$. Now we assume $b\neq0$.
By Remark \ref{eta}, if $y\in\F_{p^n}^*$, then $y^\frac{p^n-3}{2}=y^{\frac{p^n-1}{2}-1}=\eta(y)y^{-1}$, so, for, if $x\in\gf(p^n)^\sharp$ is a solution of $\Delta(x)=b$ for fixed $b\in\gf(p^n)^*$, then $x$ satisfies
\begin{equation}\label{eqnpn-3over22}
\eta(x+1)(x+1)^{-1}+\eta(x)x^{-1}=b.
\end{equation}
We distinguish the following (disjoint) four cases:
\begin{table}[H]
\centering
\begin{tabular}{|l||c|c|c|c|c|c|}
%\hline
\hline
%inserts double horizontal lines
&$x$&$\eta(x+1)$&$\eta(x)$&$Equation (\ref{eqnpn-3over22})$\\
[0.5ex]
\hline
Case I&$x\in S_{1,1}$&1&1&$x^2+(1-\frac{2}{b})x-\frac{1}{b}=0$\\
\hline
Case II&$x\in S_{-1,-1}$&$-1$&$-1$&$x^2+(1+\frac{2}{b})x+\frac{1}{b}=0$\\
\hline
Case III& $x\in S_{1,-1}$&1&$-1$&$x(x+1)=-\frac{1}{b}$\\
\hline
Case IV&$x\in S_{-1,1}$&$-1$&1&$x(x+1)=\frac{1}{b}$\\
\hline
\end{tabular}
\end{table}

%Case I. $x\in S_{1,1}$, i.e., $\eta(x+1)=\eta(x)=1$. Then (\ref{eqnpn-3over22}) becomes $x^2+(1-\frac{2}{b})x-\frac{1}{b}=0$, which has at most two solutions.
%
%Case II. $x\in S_{-1,-1}$, i.e., $\eta(x+1)=\eta(x)=-1$. Then (\ref{eqnpn-3over22}) becomes $x^2+(1+\frac{2}{b})x+\frac{1}{b}=0$, which has at most two solutions.
%
%Case III. $x\in S_{1,-1}$, i.e., $\eta(x+1)=1, \eta(x)=-1$. Then (\ref{eqnpn-3over22}) becomes $x(x+1)=-\frac{1}{b}$,  which has at most two solutions.
%
%Case IV. $x\in S_{-1,1}$, i.e., $\eta(x+1)=-1, \eta(x)=1$. Then (\ref{eqnpn-3over22}) becomes $x(x+1)=\frac{1}{b}$,  which has at most two solutions.

If $p^n\equiv 1 \pmod 4$, then $\eta(-1)=1$. If $b$ is a square element, then $\eta(-\frac{1}{b})=\eta(\frac{1}{b})=1$. In both Cases III and IV, if $x$ is a solution, then $\eta(x(x+1))=-1$. So (\ref{eqnpn-3over22}) has no solution in $S_{1,-1}$ and $S_{-1,1}$, hence (\ref{eqnpn-3over22}) has at most $4$ solutions in $\gf(p^n)^\sharp$. If $b$ is a nonsquare element, then $\eta(-\frac{1}{b})=\eta(\frac{1}{b})=-1$, and we consider the solutions in each case. In Case~I, the product of two solutions of equation $x^2+(1-\frac{2}{b})x-\frac{1}{b}=0$ is $-\frac{1}{b}$, which is a nonsquare element, this means (\ref{eqnpn-3over22}) has at most one solution in $S_{1,1}$. Similarly, we can prove that (\ref{eqnpn-3over22}) has at most 1 solution in $S_{-1,-1}$. Now we consider Case~III.
Let $x_3$ and $-x_3-1$ be the two solutions of the quadratic equation $x(x+1)=-\frac{1}{b}$. It is easy to check that $x_3\in S_{1,-1}$ if and only if  $-x_3-1\in S_{-1,1}$. This implies that (\ref{eqnpn-3over22}) has at most $1$ solution in $S_{1,-1}$. Similarly, we can prove that (\ref{eqnpn-3over22}) has at most $1$ solution in $S_{-1,1}$. We thus proved that \eqref{eqnpn-3over22} has at most 4 solutions in $\gf(p^n)^\sharp$ for $b\in\gf(p^n)^*$.

It is easy to see that $\Delta(0)=1$ and $\Delta(-1)=-1$. When $b=1$ is a square element, it was shown that $\Delta(x)=1$ has no solutions in Cases~III and~IV. In Case~I, the quadratic equation is $x^2-x-1=0$, $x=\frac{1\pm\sqrt{5}}{2}$, $x+1=\frac{3\pm\sqrt{5}}{2}=x^2$. If there are solutions, $\eta(\frac{1\pm\sqrt{5}}{2})=1$. In Case II, the quadratic equation is $x^2+3x+1=0$, $x=\frac{-3\pm\sqrt{5}}{2}$, $x+1=\frac{-1\pm\sqrt{5}}{2}$, $x=-(x+1)^2$. If there are solutions, $\eta(\frac{-1\pm\sqrt{5}}{2})=-1$, i.e., $\eta(\frac{1\pm\sqrt{5}}{2})=-1$, since $\eta(-1)=1$. Therefore, (\ref{eqnpn-3over21}) has at most $2$ solutions in $S_{1,1}$ and $S_{-1,-1}$. Then $\delta(1)\leq3$. We can prove $\delta(-1)\leq3$ in a similar way. By the above discussions, we conclude that $_c\Delta_F \leq 4$ if $p^n\equiv 1 \pmod 4$.

If $p^n\equiv 3 \pmod 4$, then $\eta(-1)=-1$. Since  $d$ is then an even number,  if $x$ is a solution of $\Delta(x)=(x+1)^d+x^d=b$ for some $b$, so is $-x-1$. This means that the number of solutions for $\Delta(x)=b$ in $S_{1,1}$ and $S_{-1,-1}$ are the same. If $b$ is a square element, (\ref{eqnpn-3over22}) has at most $1$ solution in $S_{1,1}$ (Case I) since $\eta(-\frac{1}{b})=-1$, then (\ref{eqnpn-3over22}) has at most $1$ solution in $S_{-1,-1}$ (Case II). Now we consider Case~III.
Let $x_3$ and $-x_3-1$ be the two solutions of the quadratic equation $x(x+1)=-\frac{1}{b}$. It is easy to check that $x_3\in S_{1,-1}$ if and only if  $-x_3-1\in S_{-1,1}$. This implies that (\ref{eqnpn-3over22}) has at most $1$ solution in $S_{1,-1}$.
%Similarly, we can prove that (\ref{eqnpn-3over22}) has at most $1$ solution in $S_{-1,1}$ (Case IV).
Finally, for Case IV, we have that $\eta(x(x+1))=-1$, while $\eta(\frac{1}{b})=1$, so there are no solutions in $S_{-1,1}$. We thus proved that \eqref{eqnpn-3over22} has at most 3 solutions in $\gf(p^n)^\sharp$ for $b\in\gf(p^n)^*$.

%In both Cases III and IV, if $x$ is a solution, then $\eta(x(x+1))=-1$. There are no solutions in $S_{-1,1}$, since the solution satisfies $\eta(x(x+1))=-1$. Then (\ref{eqnpn-3over21}) has at most $4$ solutions in $\gf(p^n)^\sharp$, when $b$ is a square element.

If $b$ is a nonsquare element, $\eta(\frac{1}{b})=-1$ and $\eta(-\frac{1}{b})=1$. Since in Cases III and IV we have that $\eta(x(x+1))=-1$ for any solution $x$ of (\ref{eqnpn-3over22}), there are no solutions in $S_{-1,1}$ (Case III), but there may be solutions in $S_{1,-1}$ (Case IV). However, if $x_4$ and $-x_4-1$ are the two solutions of the quadratic equation $x(x+1)=\frac{1}{b}$, it is easy to check that $x_4\in S_{1,-1}$ if and only if  $-x_4-1\in S_{-1,1}$. This implies that (\ref{eqnpn-3over22}) has at most $1$ solution in $S_{1,-1}$. As before, we can similarly prove that (\ref{eqnpn-3over22}) has at most $1$ solution in $S_{1,1}$, and at most $1$ solution in $S_{-1,-1}$. Then we proved that (\ref{eqnpn-3over22}) has at most 3 solutions in $\gf(p^n)^\sharp$, for $b\in\gf(p^n)^*$.

It is easy to see that $\Delta(0)=\Delta(-1)=1$. Now we focus on $b=1$, which is a square element. It was proven that $\Delta(x)=1$ has at most $1$ solution in
% $S_{1,1}$, at most $1$ solution in $S_{-1,-1}$, and no solution in  $S_{-1,1}$
any of $S_{1,1},S_{-1,-1}$ and $S_{1,-1}$. If $\widetilde{x}\in S_{1,-1}$ is a solution of  (\ref{eqnpn-3over22}), then $\eta(\widetilde{x}+1)=1$, $\eta(\widetilde{x})=-1$ and $\widetilde{x}(\widetilde{x}+1)=-1$. Then $\widetilde{x}$ satisfies $(\widetilde{x}+1)^2=\widetilde{x}$, the left-hand side is a square while the right-hand side is a nonsquare, which is a contradiction. Then $\Delta(x)=1$ has no solutions in $S_{1,-1}$. That means $\delta(1)\leq4$. By the above discussions, we conclude that $_c\Delta_F \leq 4$ if $p^n\equiv 3 \pmod 4$, which completes the proof.
\end{proof}

\begin{thm}\label{over3}
Let $F(x)=x^d$ be a power function over $\gf(p^n)$, where $d=\frac{2p^n-1}{3}$ and $p^n \equiv 2  \pmod 3$. For $c\neq1$, $_c\Delta_F\leq 3$.
\end{thm}
\begin{proof}
We know that $\gcd(d,p^n-1)=1$ since $3d-2(p^n-1)=1$. For any $b\in\gf(p^n)$, we consider the equation $\Delta(x)=(x+1)^d-cx^d=b$. If $x\in\gf(p^n)$ is a solution of $\Delta(x)=b$, we let $x+1=\alpha^3$ and $x=\beta^3$. Such $\alpha,\beta\in\gf(p^n)$ exist uniquely because $\gcd(3,p^n-1)=1$. Then $\alpha$ and $\beta$ satisfy $\alpha^3-\beta^3=1$ and $b=\alpha^{3d}-c\beta^{3d}=\alpha-c\beta$. Note that $x$ is uniquely determined by $\beta$ and $\beta$ satisfies a cubic equation $(b+c\beta)^3-\beta^3=1$, which has at most $3$ solutions in $\gf(p^n)$. This implies $_c\Delta_F\leq 3$.
\end{proof}

\section{Concluding remarks}\label{sec:conclusion}

In 2020, two of us, along with Ellingsen, Felke, and Tkachenko~\cite{EFRST20} defined a new (output) multiplicative differential and the corresponding
$c$-differential uniformity. In this paper, we push further the
study initiated in~\cite{EFRST20} and continued in~\cite{SPRS20,SG20}, where the authors compute the
$c$-differential uniformity of some power functions over finite fields
(which represent an important class of functions due to their low
implementation cost in a hardware environment). Precisely, here, we
study the $c$-differential uniformity of some APN power mappings
presented in \cite{HRS99,HS97} for particular values of $c$. Our
results highlight that they have $c$-differential uniformity
for those values of $c$ higher than expected and perhaps, desired. Therefore, we believe that it
would be important to study this new notion when considering PN/APN functions and differential uniformity, in general, from now on.

\vskip.2cm
\noindent
{\bf Acknowledgments.} The authors would like to express their sincere appreciation for the reviewers' careful reading, valuable comments, and suggestions. They also sincerely thank Associate  Editor Prof. Anne Canteaut for the prompt handling of our paper and interesting comments on our work.

\end{document}